\documentclass[conference]{IEEEtran}
\IEEEoverridecommandlockouts

\usepackage[utf8]{inputenc}
\usepackage{todonotes}
\usepackage{multicol}
\usepackage{url}

\usepackage{booktabs}
\usepackage{adjustbox}
\usepackage{mathtools}

\usepackage{subfigure}
\usepackage{color}
\usepackage{amsthm}
\usepackage{breakurl}
\usepackage{etoolbox}
\usepackage{cite}
\usepackage{tabularx,ragged2e}
\newcolumntype{L}{>{\RaggedRight}X}

\usepackage[font=small,labelfont=bf]{caption} 
\usepackage{float}
\usepackage{epsfig}
\usepackage[linesnumbered,algoruled,boxed,lined]{algorithm2e}
\makeatletter
\newcommand{\removelatexerror}{\let\@latex@error\@gobble}
\makeatother

\usepackage{comment}
\usepackage{pgfplots}
\usepackage{tikz}
\usepgfplotslibrary{groupplots}
\pgfplotsset{compat=newest}

\usepackage{mathtools}
\allowdisplaybreaks

\usepackage{amsthm}
\usepackage{mathtools}
\usepackage{amsmath}
\usepackage{amsfonts}
\usepackage{amssymb}
\newtheorem{theorem}{Theorem}
\newtheorem{lemma}[theorem]{Lemma}

\newtheorem{definition}{Definition}

\begin{document}
\interfootnotelinepenalty=10000

\sloppy

\title{
Optimal Latency for Partial Synchronous BFT}

\title{\Large \bf 
VBFT: Veloce Byzantine Fault Tolerant Consensus for Blockchains
}

\author{
{\rm Mohammad M. Jalalzai}\\
The University of British Columbia
\and
{\rm Chen Feng}\\
The University of British Columbia

\and 
{\rm Victoria Lemieux}\\
The University of British Columbia
}
\maketitle
\begin{abstract}
    Low latency is one of the most desirable features of partially synchronous Byzantine consensus protocols. Existing low-latency protocols have achieved consensus with just two communication steps by reducing the maximum number of faults the protocol can tolerate (from $f = \frac{n-1}{3}$ to $f = \frac{n+1}{5}$), \textcolor{black}{by relaxing protocol safety guarantees}, or by using trusted hardware like Trusted Execution Environment. Furthermore, these two-step protocols  don't support rotating leader and low-cost view change (leader replacement), which are important features of many blockchain use cases. In this paper, we propose a protocol called VBFT which achieves consensus in just two communication steps without scarifying desirable features. In particular, VBFT tolerates $f = \frac{n-1}{3}$ faults (which is the best possible), guarantees strong safety for honest leaders, and requires no trusted hardware. Moreover, VBFT supports leader rotation and low-cost view change, thereby improving prior art on multiple axes.
   
\end{abstract}
\pagestyle{plain}

\section{Introduction}\label{Introduction}

 Byzantine faults are arbitrary faults including software bugs, malicious attacks, and collusion among malicious processors (nodes), etc. Byzantine fault-tolerant protocols (e.g., \cite{Castro:1999:PBF:296806.296824, SBFT, Jalal-Window, Proteus1, HotStuff, ByzantineTolerantGeneralizedPaxos, Aardvark, avarikioti2021fnfbft}) achieve fault tolerance by replicating a state, which is often referred to as State Machine Replication (SMR) in the literature. BFT-based protocols have been used in intrusion-tolerant services \cite{Fault-Tolerance-System, Upright-Services, BFT-Critical-Infrastructure} including databases \cite{BFT-DataBases}. Recently, BFT-based consensus protocols have been actively used in blockchain technology \cite{SBFT, Jalal-Window, HotStuff,jalalzai2021fasthotstuff, Proteus1, Algorand, Casper}. In all of these use cases, the underlying BFT protocol must remain efficient. In particular, the protocol response time should be short. This helps the client requests to get executed and responses returned to the client faster, thereby greatly improving the user experience.

In BFT-based consensus protocols, one of the most important aspects that affect protocol latency is the number of communication steps. This is even more important when the protocol operates in a WAN (Wide-Area-Network) environment, as the latency for each communication step might be several hundred times higher than that in a LAN (Local-Area-Network) environment. \textcolor{black}{Various studies have shown that high latency in network applications can result in customer dissatisfaction and huge losses in revenue \cite{Akamai-Latency-Report, Amazon-Latency-Report}.  
Hence, a protocol with lower latency may generate higher revenue for the stakeholders in the network by keeping the users satisfied with its performance.
}

Generally, BFT-based consensus protocols in the partially synchronous mode\footnote{There is an unknown maximum bound on the message delay.} operate in three communication steps \cite{Castro:1999:PBF:296806.296824,BFT-SMART,Aardvark} or more \cite{HotStuff,jalalzai2020hermes,jalalzai2021fasthotstuff,SBFT} during normal execution or in the absence of failure. Here,  the normal execution latency is an important metric because
the worst-case latency of partially synchronous BFT consensus protocols can be unbounded. 
Hence, 
 BFT-based consensus protocols that operate with just two communication steps are regarded as \emph{fast} BFT consensus protocols \cite{OptimisticBFT3f, FastBFT,Revisiting-Optimal-Resilience-of-Fast-Byzantine-Consensus}.
 Current solutions, however, have one or more limitations 
(as summarized in Table~\ref{Table: Protocol-Comparision}), including lower fault tolerance \cite{FastBFT,GoodCase-Latency-PODC,Revisiting-Optimal-Resilience-of-Fast-Byzantine-Consensus}, quadratic view-change complexity \cite{OptimisticBFT3f,Castro:1999:PBF:296806.296824, Kotla:2008:ZSB:1400214.1400236, SBFT, castro01practicalPhD} in terms of signature verification cost, relaxed safety\footnote{A block from honest node can only remain in the chain if more than $f+1$ honest nodes commit it. on the other hand, during a strong safety, a block will be added to the chain permanently if only one honest node commits it.} \cite{Kotla:2008:ZSB:1400214.1400236}, and the use of Trusted Execution Environment (TEE) \cite{MinBFT}.
 
 In this paper, we propose a fast BFT consensus protocol called VBFT that improves prior art in multiple dimensions. First, VBFT can achieve consensus during normal execution with $n=3f+1$ nodes. That is, it tolerates the {maximum} possible number of Byzantine faults ($f = \frac{n-1}{3}$), thereby providing the best possible resilience achieved by a partially synchronous  Byzantine protocol \cite{Fischer:1985:IDC:3149.214121}. 
 Second, VBFT does not require the use of any trusted hardware (e.g., TEEs). 
 Third, VBFT has linear authenticator (signature) verification complexity, which allows it to change the leader efficiently. 
 Fourth, VBFT has strong safety guarantees when the leader is honest. This means if a single honest node commits a block of an honest leader, it is guaranteed that  all other nodes will commit the same block (at the same sequence in the chain) eventually.  
 The last two properties make VBFT well-suited for blockchain use cases. 
 For example, the leader rotation not only makes the protocol egalitarian (by providing an opportunity for every node to be chosen as a leader),
 but also provides a natural defense against Denial-of-Service (DoS) attacks when combined with random leader selection mechanism. Also, providing strong safety for an honest block will make sure that an honest leader receives a reward for its proposal (as the proposal won't be revoked if it is committed by at most $f$ honest nodes).

 A question arises naturally: what is the price paid by VBFT in order to achieve the above desirable properties? It turns out that VBFT only guarantees relaxed safety (rather than strong safety) for a block proposed by a Byzantine leader. Specifically, a Byzantine leader may revoke its latest block during view change (when it is replaced) if the block is committed by less than $f+1$ honest nodes. We believe that this price is acceptable  for the following reasons.

 First, even if a Byzantine leader revokes its latest block, it cannot create any double spending\footnote{A request is considered committed by a client, which is then later revoked.}, because this block will not be considered committed by any clients.
 Second,  revoking a block proposed by a Byzantine leader will not lower the chain quality\footnote{Chain quality is defined as the fraction of blocks added by honest leaders in the blockchain.}, which is an important metric related to fairness. By contrast, revoking a block proposed by an honest leader may lower the chain quality. In other words, we will not reduce fairness by only guaranteeing strong safety for honest leaders.
 Third, a block revocation is only possible if the leader node performs equivocation (proposing conflicting blocks). 
 In reality, equivocation attacks can easily be detected . Once an equivocation is detected, the culprit leader can be blacklisted and so no equivocation will take place from this leader in the future. 
\textcolor{black}{In other words, the cost of performing equivocation by a Byzantine leader is so high that it is not in its best interest to get blacklisted and not be able to participate in consensus in the future. Moreover, if this Byzantine leader has a stake in the network, it will lose its stake\footnote{In the Proof-of-Stake incentive mechanism, nodes have a stake in the network in the form of tokens. If a node behaves maliciously then its stake will be slashed. But a node is rewarded when it participates in consensus (proposes a block that is eventually added to the chain).} due to such misbehavior.}

To sum up, our VBFT protocol exhibits the following desirable properties:
\begin{itemize}
    \item Optimal latency during normal protocol execution
    \item Maximum fault tolerance
    \item Strong safety guarantee for honest leader nodes
    \item Efficient view change
    \item Rotating leader paradigm
    \item Simple by design
\end{itemize}

\begin{figure}
    \centering
    \includegraphics[width=6cm,height=4cm]{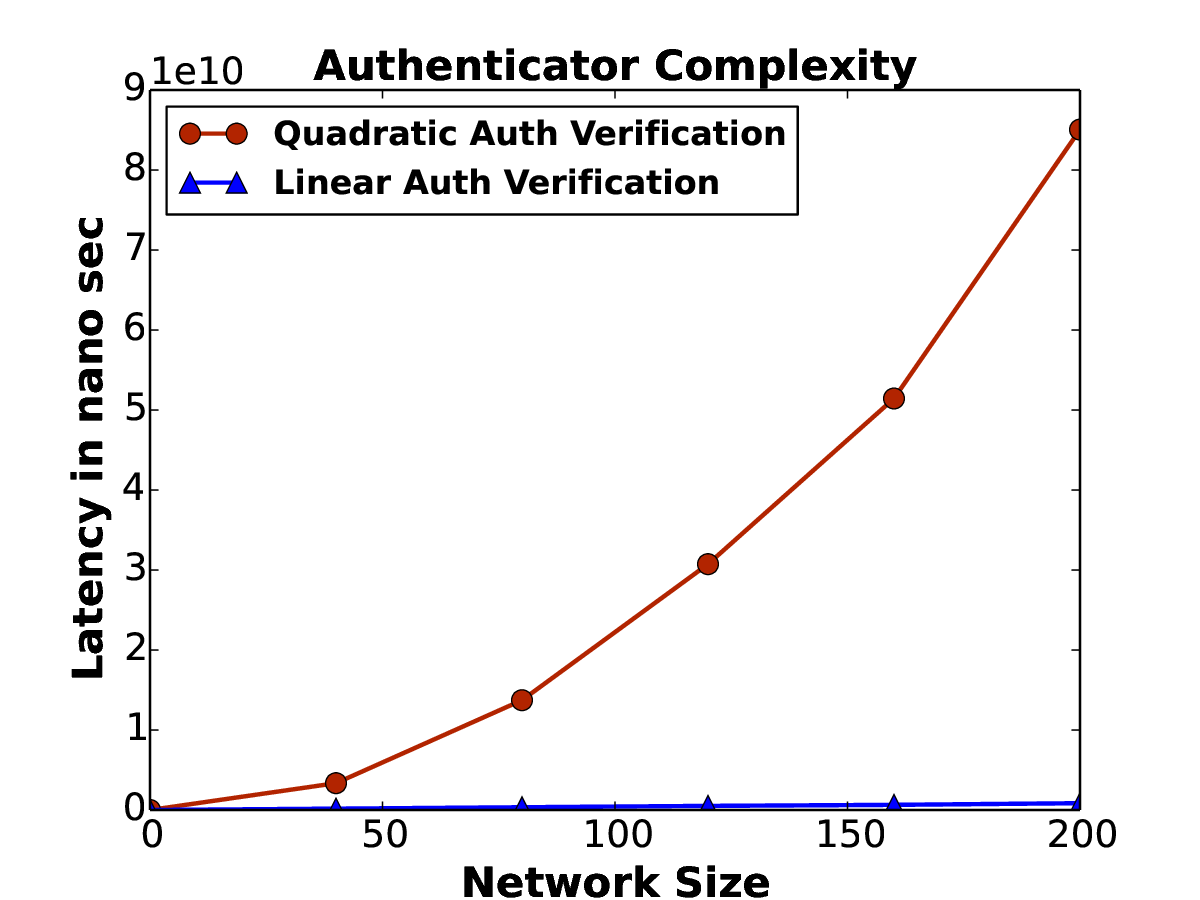}
    \caption{\textbf{Linear vs Quadratic Authenticator Complexity}.}
    \label{fig:Auth-Complexity}
\end{figure}

The remainder of this paper is organized as follows:
Section \ref{sec:overview} provides an overview of VBFT and compares it with prior art.
Section \ref{Section: System Model} presents the system model, definitions, and preliminaries. Section \ref{Section: Protocol} presents  protocol details. 
Section \ref{Section: Proof} provides formal proofs for VBFT. 
Section \ref{Section: Conclusion} concludes the paper.

\section {Overview and Prior Work}\label{sec:overview}

\subsection{Overview}

Our VBFT protocol has been optimized over several properties. First, it achieves consensus in just two communication steps while providing optimal fault tolerance. Fault tolerance improves the protocol resilience in comparison to the fast BFT protocols that have trade fault tolerance for lower latency. Another important property that VBFT has improved is the view change  (when a leader is replaced) performance. View change authenticator complexity has been reduced by $O(n)$ without using TEE or sacrificing fault tolerance. Figure \ref{fig:Auth-Complexity} shows the time taken to verify aggregated BLS signatures for different network sizes (with linear vs quadratic authenticator complexity). From Figure \ref{fig:Auth-Complexity} we can see how $O(n)$ signature verification cost of VBFT will help it to achieve efficient view change in comparison to the other protocols with quadratic authenticator verification \cite{castro01practicalPhD, Castro:1999:PBF:296806.296824, SBFT}. 

Safety is an important property of consensus protocols. We informally define two types of safety here, namely, strong safety and relaxed safety. Strong safety means when a single honest node commits a value, all other honest nodes will eventually commit the same value. In a protocol with relaxed safety a value will never be revoked (will be committed by all honest nodes eventually) if $f+1$ honest nodes commit it. This means if a value is committed by less than $f+1$ honest nodes, it is possible that the value is revoked due to the network asynchrony or malicious behavior. If the revoked value was proposed by an honest node, then it reduces the chain quality. Lower chain quality also allows Byzantine nodes to delay honest transactions. This also means that blockchain now contains a higher percentage of proposals from Byzantine leaders/leaders.

The chain quality is a property that also affects the incentive mechanism and eventually the security of a consensus protocol. In BFT consensus, chain quality becomes even more important when the votes are weighted based on the amount of stake each node (validator) holds. In this case, the assumption is that the Byzantine nodes must hold less than one-third of the total stake of the network. Therefore, a proposal will be accepted once a sufficient number of votes for it is received such that the total stake of voters is more than two-thirds of the network. Hence, when the chain quality of a protocol degrades, the number of Byzantine blocks is more in the chain than it needs to be. Since Byzantine nodes get a reward for adding a block to the chain, the rate at which the stake of Byzantine nodes increases will be faster than the rate of honest nodes.  Once the stake of Byzantine nodes gets more than one-third of the total stake, the protocol cannot guarantee liveness (cannot make progress).
Therefore, a consensus protocol needs to make sure the chain quality is not degraded. In other words, honest nodes should be able to add their blocks to the blockchain as much as possible. Therefore, relaxed safety can be exploited to reduce the chain quality and eventually the protocol security in PoS-based BFT consensus.

Similarly, as we stated before, rotating leader allows the protocol to be fair and improves defense against DoS attacks. An incentive mechanism cannot be implemented if a BFT protocol cannot provide a fair opportunity to every node to successfully add its proposal to the chain. \textcolor{black}{Below we present the limitations of prior work as summarized in Table \ref{Table: Protocol-Comparision}.}

\subsection{Prior Work}

\textbf{Lower Fault Tolerance.} Fast-BFT \cite{FastBFT} can achieve consensus in two communication steps, but the maximum number of Byzantine faults it can tolerate during the worst case is small $f = \frac{n-1}{5}$. Another recent proposal \cite{GoodCase-Latency-PODC} has further improved the results of \cite{FastBFT} and shown that Byzantine consensus can be achieved in two communication steps where maximum faults tolerated is $f = \frac{n+1}{5}$. 
A concurrent but similar work to \cite{GoodCase-Latency-PODC} is presented in  \cite{Revisiting-Optimal-Resilience-of-Fast-Byzantine-Consensus}, where the authors revisit optimal resilience for fast Byzantine consensus and present a tight lower bound for the resiliency $n = 5f-1$ (where $f = \frac{n+1}{5}$). VBFT improves this lower bound by relaxing the strong safety for Byzantine leader.
Zyzyvya achieves consensus in one communication round if all nodes are honest. This means the fault tolerance in this case is zero. Upon receipt of a proposal from the leader (leader), a node speculatively executes it hoping that every other node will also successfully execute the command. But it will fall back to a three communication step protocol if even one node out of $n$ is Byzantine or may not deliver the message on time. Moreover, Zyzvya depends on clients (trusts clients) for message propagation, making it unsuitable for the blockchain use case.

\textbf{View Change Complexity.} The computation cost of signature verification during the view change is one of the main causes of the bottleneck (during the view change) in the case of blockchain \cite{jalalzai2021fasthotstuff}. 
Reducing this cost can positively improve the view change performance. In PBFT\cite{Castro:1999:PBF:296806.296824}, Zyzyvya \cite{Kotla:2008:ZSB:1400214.1400236}, PBFT with the tentative execution \cite{castro01practicalPhD} and \cite{OptimisticBFT3f}, view change signature verification cost is quadratic. Similar to Fast-HotStuff \cite{jalalzai2021fasthotstuff}, VBFT has a linear authenticator verification cost. Unlike Fast-HotStuff (which needs more than two communication steps), VBFT only needs two communication steps to achieve consensus.

\textbf{Safety.}
\textbf{}\textcolor{black}{ 
PBFT with tentative execution \cite{castro01practicalPhD}
can achieve consensus in two communication steps, but it can only provide relaxed safety guarantees. In a protocol with a relaxed safety guarantee, a request is committed permanently only if at least $f+1$ (compared to $1$ honest node during strong safety) honest nodes commit it. \textcolor{black}{In PBFT with tentative execution, a request execution can be revoked, even if the leader (leader) is not malicious. As discussed, the revocation of honest blocks results in chain quality degradation. This can have negative effects on protocol security when an incentive mechanism is deployed over consensus.}}

\textbf{Rotating leader.}
Previous fast BFT consensus protocols do not support rotating leader feature \cite{Kotla:2008:ZSB:1400214.1400236, castro01practicalPhD, FastBFT, MinBFT, GoodCase-Latency-PODC}. Enabling rotating leader is desirable for blockchain use cases, as it allows fairness and improves resilience to DoS attacks.

\textbf{Trusted Execution Environment.}
MinBFT \cite{MinBFT} uses trusted hardware TEEs (Trusted Execution Environment) to prevent equivocation attack \footnote{An attack during which a malicious leader/leader proposes multiple blocks for the same height.}. 
Although TEEs are considered to be secure, some recent work (e.g., \cite{TEE-Vulnerability}) raises legitimate concerns about their security by identifying several vulnerabilities.


\begin{figure}
\footnotesize
    \captionof{table}{Comparison of Fast BFT Protocols}
\begin{tabularx}{\textwidth}{c@{\hspace{\tabcolsep}}cccccc}
 \cline{1-7}
\textbf{Protocol} &   \multicolumn{1}{p{0.4cm}}{\centering  
\textbf{CS}} &   \multicolumn{1}{p{0.4cm}}{\centering \textbf{FT}} &   \multicolumn{1}{p{0.4cm}}{\centering \textbf{SS-HP}} &   \multicolumn{1}{p{0.4cm}}{\centering \textbf{AV}} & \multicolumn{1}{p{0.4cm}}{\centering  \textbf{TEE}}  & \multicolumn{1}{p{0.4cm}}{\centering  \textbf{RP}}  \\\cline{1-7}

 Fast BFT\cite{FastBFT} & 2  & (n-1)/5 & Yes & $O(n)$ &  No & No \\ 
 Psync-BB\cite{GoodCase-Latency-PODC} & 2  & (n+1)/5 & Yes & $O(n)$ & No & No \\ 
 
  Min\cite{MinBFT} & 2  & (n-1)/3 & Yes & $O(n)$ & Yes & No \\ 
   PBFT-TE\cite{castro01practicalPhD} & 2  & (n-1)/3 & No & $O(n^2)$ & No & No \\ 
   { Zyzyvya \cite{Kotla:2008:ZSB:1400214.1400236}} & {1-3}  & {0 or (n-1)/3} & {No} & {$O(n^2)$} & {No} & No\\ 
  \textbf{ VBFT (this work)} & \textbf{2}  & \textbf{(n-1)/3} & \textbf{Yes} & \textbf{$O(n)$} & \textbf{No } & \textbf{Yes}\\\cline{1-7}
\label{Table: Protocol-Comparision}
\end{tabularx}
\caption{Comparision of VBFT with the PBFT and Fast BFT variants. Here, CS stands for communication steps, FT stands for fault tolerance, SS-HP stands for strong safety for honest leaders, AV stands for authenticator verification cost and RP for rotating leader paradigm.}
\end{figure}

\section{Definitions and Model}
\label{Section: System Model}
   VBFT operates under the Byzantine fault model. VBFT can tolerate up to $f$ Byzantine nodes where the total number of nodes in the network is $n$ such that $n=3f+1$. The nodes that follow the protocol are referred to as correct/honest nodes. 
 
   Nodes are not able to break encryption, signatures, and collision-resistant hashes. We assume that all messages exchanged among nodes are signed.

To avoid the FLP impossibility result \cite{Fischer:1985:IDC:3149.214121}, VBFT assumes a partial synchrony model \cite{Dwork:1988:CPP:42282.42283} with a fixed but unknown upper bound on the message delay.
   The period during which a block is generated is called an epoch. Each node maintains a timer. If an epoch is not completed within a specific period (called timeout period), then the node will time out and trigger the view change (changing the leader) process. The node then doubles the timeout value, to give enough chance for the next leader to drive the consensus. 

   \subsection{Preliminaries}
\noindent \textbf{View and View Number.}
In a distributed system, a \textit{view} is a particular configuration of nodes, where one node is designated as the leader and proposes a set of transactions to be executed. A \textit{view number} is a monotonically increasing identifier associated with a specific view. Each node in the system uses a deterministic function to map a view number to a node ID, which uniquely identifies the node within the system. Specifically, given a view number $v$ and a set of nodes $N = {n_1, n_2, ..., n_k}$, each node $n_i \in N$ uses a function $f_v$ to map $v$ to a node ID $id_i = f_v(v)$, such that the leader node for view $v$ is identified by the node with ID $id_i = f_v(v)$ in the set $N$.
 
 \noindent \textbf{Signature Aggregation.} VBFT uses signature aggregation \cite{Short-Signatures-from-the-Weil-Pairing,Boneh:2003,BDNSignatureScheme} in which signatures from nodes can be aggregated into a 
 a single collective signature of constant size. 
 Upon receipt of the messages ($M_1,M_2,\ldots, M_y$ where $ 2f+1 \leq y \leq n$)   with their respective signatures ($\sigma_1,\sigma_2,\ldots, \sigma_y$)
 the leader then generates an aggregated signature $\sigma \gets AggSign(\{M_i, \sigma_i\}_{i \in N})$. The aggregated signature can be verified by nodes given the messages $M_1,M_2,\ldots, M_y$ where $ 2f+1 \leq y \leq n$, the aggregated signature $\sigma$, and public keys $PK_1,PK_2,\ldots,PK_y$. Signature aggregation has previously been used in BFT-based protocols \cite{Castro:1999:PBF:296806.296824,Lamport:1982:BGP:357172.357176, Jalal-Window}. Any other signature scheme may also be used with VBFT as long as the identity of the message sender is known.

 \noindent \textbf{Block and Blockchain}
 Requests are also called transactions. In a blockchain, transactions are batched into blocks for efficiency. Each block contains a set of transactions and a hash pointer to the previous block, forming a chain of blocks called a blockchain.
 
 \noindent \textbf{Quorum Certificate (QC), Timeout Certificate (TC) and No-Commit Certificate (NC).}
 A quorum certificate (QC) is a set of $2f+1$ votes for a specific view from distinct nodes in a blockchain network. For instance, $QC$ is formed by collecting at least $2f+1$ votes from different nodes for a block. A block is considered certified when it receives enough votes to build its $QC$ or when the $QC$ itself is received. Similarly, a $TC$ is the collection of $2f+1$ \textsc{timeout} messages. The $QC$ with the highest view in the $TC$ is called highQC . An $NC$ (Negative Certificate) is crafted from $2f+1$ negative response messages. When a leader inquires if a node possesses the payload for a particular block header and the node does not possess said payload, it issues a negative response message in reply.

 \noindent \boldmath{\textbf{Unverified Block Header} ($U$).  Represents the header of a block for which a node has voted but has not yet received a $QC$. During view changes, the unverified block header is used to ensure that a block committed by fewer than $f+1$ nodes remains in the chain. 
 
   \subsection{Definitions}

\begin{definition}[Safety]
A protocol is considered \textbf{safe} in the presence of $f$ Byzantine nodes and under partial synchrony if, for any block $b$ committed during view $v$, the following conditions hold:
\begin{itemize}
     \item $b$ is the only block committed during view $v$.
     \item There exists no other block, denoted as $b^{'}$, that is committed during any view $v^{'}$ greater than $v$ and $b^{'}$ does not extend $b$ or act as a predecessor to $b$.
\end{itemize}
In simpler terms, the safety property of consensus guarantees that once a block 'b' is committed during view $v$, no conflicting (does not extend $b$ or precede $b$) block $b^{'}$ will ever be committed to any subsequent view $v^{'}$, and $b$ remains the only committed block during view $v$.
\end{definition}

\begin{definition}[Strong Safety]
A protocol is considered \textbf{S-safe} in the presence of $f$ Byzantine nodes and under partial synchrony if safety holds for any block $b$ committed by a single honest node during view $v$. It means that if a single honest node commits a block $b$ during view $v$, then eventually all other nodes will commit the same block $b$ during view $v$.
\end{definition}

\begin{definition}[Relaxed Safety]
A protocol is considered \textbf{R-safe} in the presence of $f$ Byzantine nodes and under partial synchrony if safety holds for any block $b$ committed by at least $f+1$ honest nodes during view $v$. It means that if $f+1$ honest nodes commit a block $b$ during view $v$, then eventually all other nodes will commit the same block $b$ during view $v$.
\end{definition}

  \begin{definition}[Liveness]
A protocol is considered to be \textbf{alive} if, in the presence of at most $f$ Byzantine nodes:

\begin{enumerate}
\item Each correct node eventually commits a block to the blockchain.
\item The blockchain continues to grow with additional blocks being appended to it over time.
\end{enumerate}

In other words, the protocol must guarantee progress even in the presence of up to $f$ Byzantine nodes.

\end{definition}

\SetKwFor{Upon}{upon}{do}{end}
\SetKwFor{Check}{check always for}
{then}{end}



\setlength{\textfloatsep}{1.5pt}

     \SetKwFor{Upon}{upon}{do}{end}
\SetKwFor{Check}{check always for}
{then}{end}
\SetKwFor{Checkuntil}{check always until}
{then}{end}
\SetKwFor{Continuteuntil}{Continue until}
{then}{end}
\begin{algorithm}
\DontPrintSemicolon
\SetAlgoLined

\caption{Utilities for Node $i$}
\label{Algorithm:Utilities}

\SetKwFunction{CreatePrepareMsg}{CreatePrepareMsg}
\SetKwFunction{BQC}{GenerateQC}
\SetKwFunction{CreateTC}{CreateTC}
\SetKwFunction{PipelinedSafeBlock}{SafetyCheck}
\SetKwProg{Fn}{Func}{:}{}

\Fn{\CreatePrepareMsg{type, v, s, d, hash, $NC$, $\Sigma$}}{
    $b.\text{type} \gets \text{type}$\;
    $b.v \gets v$\;
    $b.\text{block} \gets \text{hash}$\;
    $b.\text{parent} \gets d$\;
    $b.QC \gets \text{qc}$\;
    $b.NC \gets NC$\;
    $b.\Sigma \gets \Sigma$\;
    \Return $b$\;
}

\Fn{\BQC{$V$}}{
   $qc.\text{viewNumber} \gets m.\text{viewNumber}$ for $m \in V$\\
$qc.\text{block} \gets m.\text{block}$ for $m \in V$\\
$qc.\text{sig} \gets \text{AggSign}(qc.\text{type}, qc.\text{viewNumber},$ \\
$\quad \quad qc.\text{block}, i, \{m.\text{Sig} \,|\, m \in V\})$\;

    \Return $qc$\;
}

\Fn{\CreateTC{$timeout_{Set}$}}{
    $TC.\text{QCset} \gets$ extract QCs from $timeout_{Set}$\;
    $TC.\text{sig} \gets \text{AggSign}(curView, \{qc.\text{block} \,|\, qc.\text{block} \in TC.\text{QCset}\}, \{U.\text{block} \,|\, U.\text{block} \in \mathcal{U}\}, \{i \,|\, i \in N\}, \{m.\text{Sig} \,|\, m \in timeout_{Set}\})$\;
    $TC.\text{highQC} \gets \arg \max_{QC \in QCs} (QC.\text{v})$\;
    $TC.\text{highU} \gets \{ U \in \mathcal{U} \,|\, U.\text{parent} = \text{TC.highQC.block} \}$\;
    \Return $TC$\;
}

\Fn{\PipelinedSafeBlock{$b$, $qc$, $tc$}}{
    \If{$b.v == qc.v+1$}{
        \Return \textcolor{blue}{b extends highQC.block}\;
    }
    \If{$b.v > qc.v+1$}{
        \If{$tc.highU.s == qc.s+1$}{
            \If{TC.NC}{
                \Return \textcolor{blue}{b extends highQC.block $\land b.v == tc.v+1$}\;
            }
            \Else{
                \Return \textcolor{blue}{$b.U == tc.highU \land$ b extends highQC.block}\;
            }
        }
        \Else{
            \Return \textcolor{blue}{b extends highQC.block $\land b.v == tc.v+1$}\;
        }
    }
}
\end{algorithm}


\section{VBFT Protocol} 
\label{Section: Protocol}
The VBFT protocol operates in two modes, namely the normal and view change mode. The protocol operates in the normal mode where blocks are added into the chain until a failure is encountered. To recover from failure, VBFT switches to the view change mode. Upon successful completion of the view change mode, VBFT again begins the normal mode.

\subsection{Happy Path Execution}

\begin{algorithm*}
\SetKwFunction{CreatePrepareMsg}{CreatePrepareMsg}
\SetKwFunction{SafeProposal}{SafeProposal}
\SetKwProg{Fn}{Function}{:}{}
\small

\caption{Normal Execution for Node $i$}
\label{Algorithm: Normal-Mode}

\ForEach{curView $\gets$ 1, 2, 3, ...}{
    \If{i is leader}{
        \Upon{Receipt of 2f+1 \textsc{timeout} messages}{
            \If{There is $U.v == highQC.v+1$}{
                \If{Payload ($\Sigma$) for $U$ is present $\lor$ Receipt of payload for $U$}{
                    $B\gets$ \CreatePrepareMsg{Prepare, $v$, hash, $parent$, $qc$, $\perp$, $tc$, $\Sigma$ (from $U$)}\;
                }
                \Else{
                    Broadcast Request for $U$ payload\;
                    \Upon{Receipt of $2f+1$ Negative Response messages}{
                        $B\gets$ \CreatePrepareMsg{Prepare, $v$, hash, $parent$, hash, $qc$, $NC$, $\perp$, $\Sigma$}\;
                    }
                }
            }
            \Else{
                $B\gets$ \CreatePrepareMsg{Prepare, $v$, hash, $parent$, $qc$, $\perp$, $tc$, $\Sigma$}\;
            }
        }
        \Upon{Receipt of $2f+1$ votes}{
            $B\gets$ \CreatePrepareMsg{Prepare, $v$, hash, $parent$, $qc$, $\perp$, $\perp$, $\Sigma$}\;
        }
        Broadcast $B$\;
    }
    \If{i is a normal node}{
        \Upon{Receipt of Request for $U$ payload}{
            Send payload to the leader or Negative Response if payload is not available\;
        }
        \text{Wait for $B$ from leader (curView)}\;
        \If{SafetyCheck($B, B.QC, TC |\perp$) $\land b.v > curView$}{
            Broadcast vote $\langle v, s, h, parent, i \rangle$ for prepare message\;
            Increment $curView$ to $B.QC.v+1$\;
        }
        \Upon{Receipt of $2f+1$ votes for block $B$}{
            Execute new commands in $B$\;
            Respond to clients\;
        }
    }
    \Check{For timeout}{
        Increment $curView$ to $B.QC.v+1$\;
        Broadcast $\textsc{timeout}$ message\;
    }
}
\end{algorithm*}

In the context of the normal operational mode, we present a structured message pattern, as illustrated in Figure~\ref{fig:Message-pattern}. During this mode, the leader node is entrusted with the reception of  transactions originating from clients. Each client, denoted as $c$, transmits its signed transaction, conforming to the format $\langle REQUEST, o, t, c \rangle$, to all nodes within the network. Here, $o$ signifies the requested operation in the form of transaction by client $c$, and $t$ serves as a timestamp utilized by the leader to orderly sequence requests originating from said client.

The leader node, designated for the current view $v$, undertakes the creation of a proposal $b$, interchangeably referred to as a 'block' or 'pre-prepare message.' This proposal message is formally defined as $\langle \textsc{Pre-prepare}, v, block,parent,\textcolor{black}{qc}, tc, nc, \Sigma \rangle$, where it incorporates essential elements such as the view number $v$,and the executable transactions ($\Sigma$).

Notably, block and \textcolor{black}{$qc$ represent  the hash of the block and the Quorum Certificate ($QC$). The $QC$ is constructed from aggregated votes about the parent block. The block header, denoted as $U$, is succinctly expressed as $\langle \textsc{Pre-prepare}, v, block,parent \rangle$. It is imperative to observe that the inclusion of the $QC$ for the parent block is delegated to the subsequent rotating leader, a measure taken to ensure the proposal of the current block only occurs after the majority of nodes have extended their votes in favor of the parent block. $TC$ signifies the timeout certificate originating from the previous view and finds application exclusively in scenarios where the preceding view fails to certify a block, commonly referred to as the 'unhappy path'.}

$NC$ represents the certificate constructed from 'no-commit' messages. During the proposal of the first block following a failure, the newly appointed leader may seek to validate the existence of a block committed by a maximum of $f$ nodes, assuming that the $QC$ of said block remains unincorporated within $TC$. In the absence of such a block, the leader aggregates 'no-commit' messages from $2f+1$ nodes, subsequently forming an $NC$. Consequently, the $TC$ and $NC$ fields are intentionally designed for one-time use, exclusively in the context of system recovery following a failure or in the aftermath of traversing the 'unhappy path'. Detailed elaborations on this matter will be reserved for a dedicated discussion within the Unhappy Path Execution subsection (refer to subsection \ref{Subsection: Unhappy Path} for comprehensive insights).

The leader node disseminates the proposal to all network nodes during the pre-prepare phase, as delineated in Algorithm \ref{Algorithm: Normal-Mode}, line 23. Upon receiving a proposal, each individual node performs a  safety check through the employment of the $SafetyCheck()$ predicate.

The fundamental purpose of the $SafetyCheck()$ predicate is to ascertain that the proposed block $b$ extends from the most recently committed block. In the context of the 'happy path,' the first condition is met when ($b.v == qc.v + 1$), indicating the certification of the parent block ($b$'s parent). Conversely, in scenarios involving system recovery from an 'unhappy path,' the second condition ($b.v > qc.v + 1$) is satisfied, signifying the absence of certification for block $b$'s parent. It is crucial to underscore that our current discussion predominantly focuses on the 'happy path,' while a comprehensive exploration of the 'unhappy path' is reserved for subsequent subsections.

Upon the successful completion of the safety check, a node optimistically executes the block proceeds to send its vote in the format of $\langle Vote, v, block, state\_hash,\text{parent} \rangle$ to all other nodes in the network including the next leader.  The $state\_hash$ is the new state hash generated as the optimistic execution of the block.
The node affixes its signature to the vote, employing the tuple $\langle v, block,state\_hash, \text{parent} \rangle$. Having the resultant $state\_hash$ helps to detect if a node has a different execution state.

Upon the reception of $2f+1$ votes endorsing a particular block $b$ with the matching view, block hash, state hash, and the parent block, each node takes the significant step of committing the block, thereby initiating the commit phase. Additionally, each node assembles a Quorum Certificate ($QC$) based on the votes received for block $b$. This $QC$ serves a dual purpose: it acts as verifiable evidence of block commitment during view changes and fulfills the requirement for providing proof of commitment as necessitated by applications operating atop the VBFT protocol. \textcolor{black}{Furthermore, the inclusion of the parent block's $QC$ serves the strategic function of indicating that the leader intends to propose the subsequent block after the successful commitment of its parent. This measure also serves to deter subsequent leader nodes from inundating the network with redundant proposals. In this context, when a node receives a proposal, it is not required to validate the $QC$ within the proposal if it has already received $2f+1$ votes for the parent proposal within the protocol's critical path. This validation can be performed outside the critical path, if necessary, to address instances of Byzantine leadership, without compromising the protocol's overall performance.}

Additionally, each node that has successfully committed the block is tasked with dispatching a $Reply$ message \textcolor{black}{($\langle Reply, v, r, b \rangle$)} to clients whose transactions have been included in the block. Here, $r$ represents the hash of the latest state as a result of the committed block execution. Generally, it is the hash of the Merkle tree  that maintains the history of the blockchain. For a client to consider its request or transaction as committed, it must receive a minimum of $2f+1$ distinct $Reply$ messages affirming the commitment.

Likewise, upon collecting $2f+1$ votes, the newly appointed (rotating) leader assumes the responsibility of committing the block and subsequently proposing the next block in its predefined queue. In the event that a node fails to complete a view within the designated timeout period, it initiates the broadcasting of a $\textsc{timeout}$ message.

\begin{figure}
    \includegraphics[width=8cm,height=3cm]{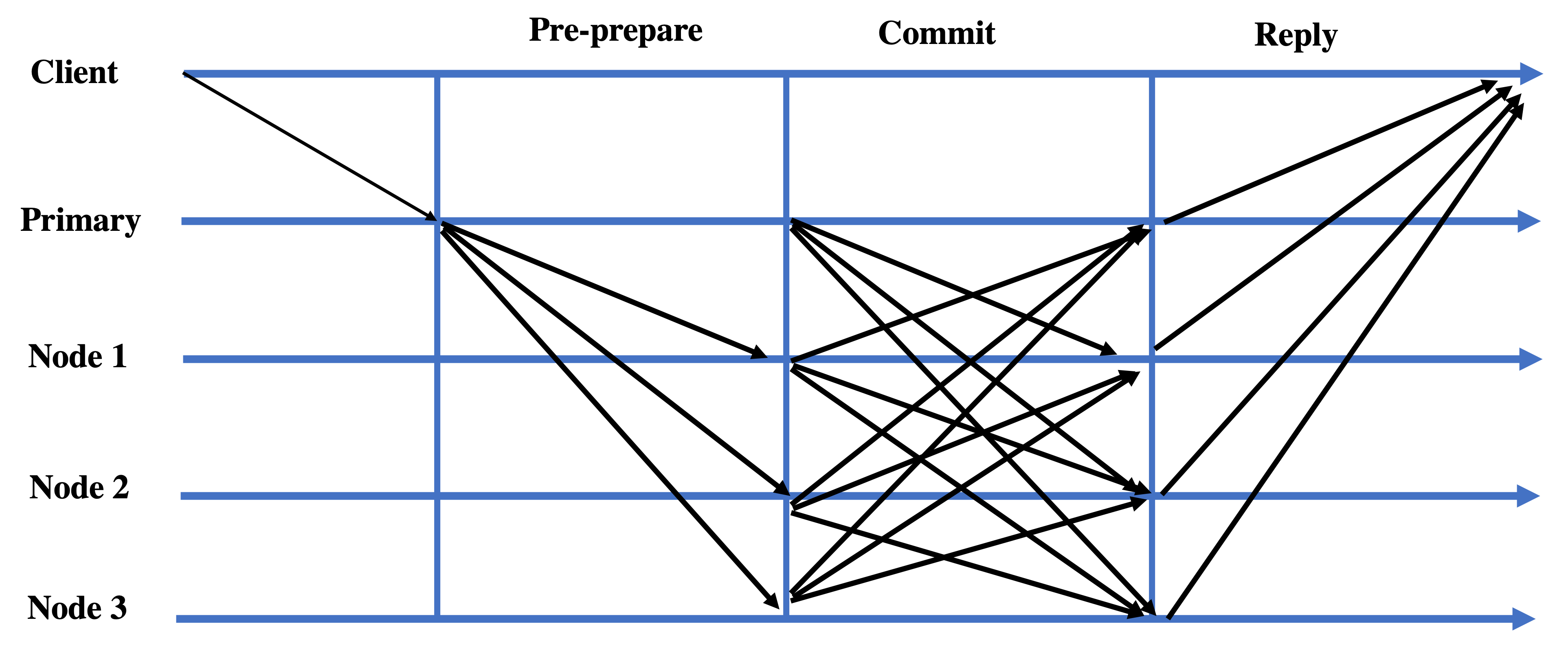}
    \caption{\textbf{VBFT Message Pattern}.}
    \label{fig:Message-pattern}
\end{figure}

\subsection{UnHappy Path Execution (Execution Failure)}   
\label{Subsection: Unhappy Path}
In the course of the 'happy path,' the succeeding rotating leader seamlessly proceeds to propose the next block once it has successfully committed the current one. Conversely, in the 'unhappy path,' a $\textsc{timeout}$ message is employed to effectuate the replacement of the failed leader with a new leader.

Node assignments to leader roles are facilitated through the utilization of a deterministic function that maps the view number to the leader ID. The selection of leaders can be orchestrated in a round-robin fashion \cite{HotStuff, Castro:1999:PBF:296806.296824, Jalal-Window} or via random selection \cite{Coin-Flipping-First-Paper, jalalzai2020hermes}.

In addition to designating a new leader, the process of view change necessitates the synchronization of information among nodes. More specifically, the newly appointed leader must ascertain the parent block of the block it intends to propose following the leader's failure. The block proposed immediately after a failure must invariably extend from the latest block that has been committed across the network \cite{SBFT, HotStuff, Castro:1999:PBF:296806.296824} (referred to as S-safety).

\textcolor{black}{As stipulated in VBFT, if a single honest node commits a block ($b$) proposed by an honest leader just before a view change, it is guaranteed that any subsequent leader must extend block $b$. This guarantee of S-safety for an honest leader is upheld while maintaining maximum resilience ($f = \frac{n-1}{3}$). 
VBFT not only guarantees R-safety for Byzantine leaders, meaning that a block proposed by a Byzantine leader and committed by at most $f$ honest nodes can be revoked if the Byzantine leader engages in equivocation just prior to the view change, but it also ensures that any Byzantine leader who is subjected to such revocation will be subsequently blacklisted. Moreover, this revocation mechanism has no detrimental impact on clients and does not lead to double-spending.}

If a node experiences a communication failure with the leader within a predefined time interval, if a command remains unexecuted, or if a node accumulates $f+1$ \textsc{timeout} messages for a particular view, it takes the initiative to broadcast a $\textsc{timeout}$ message $\langle \textsc{timeout}, v+1, qc, U\rangle$ (as described in Algorithm \ref{Algorithm: Normal-Mode}, lines 39-41). In this context, $qc$ signifies the Quorum Certificate ($QC$) for the most recently committed block, while $U$ represents the header of the latest uncommitted block for which the node has cast its vote.

Upon the reception of $2f+1$ \textsc{timeout} messages, the messages are aggregated into a collective entity known as the Timeout Certificate ($TC$) by the new leader. The $TC$ is succinctly represented as $\langle \textsc{Timeout-certificate},v, qcs, \mathcal{U}, highQC,highU\rangle$, consolidating the information contained within the $2f+1$ \textsc{timeout} messages. Within this structure, $qcs$ maintains a record of the views associated with each $QC$ received from distinct nodes.

On the other hand, if a node receives 
$f+1$ \textsc{timeout} messages, it also initiates a view change and broadcasts its \textsc{timeout} message. This proactive approach accelerates the recovery process from an unhappy path, enhancing the protocol's overall efficiency \cite{Castro:1999:PBF:296806.296824}.

In a similar fashion, $\mathcal{U}$ represents a multiset where each element corresponds to the most recent block header for which a node has cast its vote. Furthermore, we introduce the term $highU$, which signifies an element $U \in \mathcal{U}$ that satisfies the condition $U.parent == highQC.block$. This condition deliberately filters out scenarios where $highU.parent$ is an ancestor of $highQC.block$, emphasizing the significance of $highQC$. Conversely, when $highU.parent$ is a predecessor of the child of $highQC.block$, it highlights an impossibility and underscores that there exists no valid $U$ committed by only one node for more than a single view.

In VBFT, preserving safety during the unhappy path is crucial. This means that after a failure the first leader selected has to extend the latest committed block.
Recovering from an unhappy path involves three primary scenarios centered around a block proposed just before a failure:

\textbf{Scenario 1: Block Committed by at Least $f+1$ Honest Nodes}

In this scenario, if a block is committed by a minimum of $f+1$ honest nodes, VBFT ensures that the protocol extends the latest committed block. 
This can be achieved by extending the block of the  $highQC$ in $TC$.
    Since if $f+1$ honest nodes commit the latest block, then any combination of $2f+1$ $\textsc{timeout}$ messages will have at least one $QC$ for the committed block by at least $f+1$ \textcolor{black}{honest} nodes (for more details please refer to Lemmas \ref{Lemma:S-safe-Normal} and \ref{Lemma: S-safety-ViewChange-NoEq}). 
       Though the $QC$ of a block committed by at most $f$ nodes may also end up in $TC$, it is not guaranteed. If it did, then the new leader will extend the block committed by at most $f$ nodes.

\textbf{Scenario 2: Block Committed by At Most $f$ Honest Nodes}

In the VBFT protocol, specific scenarios may emerge wherein a block becomes committed by no more than $f$ honest nodes. To address such situations while maintaining safety guarantees, VBFT incorporates tailored mechanisms. The protocol ensures that a block achieves commitment by at least one honest node but not more than $f$ nodes by leveraging the concept of $U$ values found within the set $\mathcal{U}$ encapsulated within the Timeout Certificate ($TC$).

When a block attains commitment by fewer than $f+1$ honest nodes, there exists the possibility that its corresponding Quorum Certificate ($QC$) may not find representation among the $2f+1$ $\textsc{timeout}$ messages enclosed within $TC$. However, owing to the participation of $2f+1$ nodes, including $f+1$ honest ones, in casting their votes for the block, it is assured that at least one $U$ value contained within $TC$ pertains to the most recently committed block. It is worth noting that an alternative scenario exists where the block or request associated with this particular $U$ value may not have achieved commitment by a single honest node. In such an eventuality, the succeeding leader is tasked with the responsibility of procuring evidence to substantiate the non-commitment status of the block linked to $U$.

In this context, if the Quorum Certificate ($QC$) corresponding to the latest committed block is present within the $TC$, it is designated as the highest Quorum Certificate ($highQC$) contained within the $TC.$ This designation allows the subsequent leader to extend the block associated with $highQC.$ Conversely, in cases where the $QC$ of the latest committed block fails to manifest within the $TC,$ the $highU$ value enclosed within the $TC$ is employed as the foundational reference for the succeeding leader to extend the block associated with $highU.$ If the leader is in possession of the payload for $TC.highU,$ it will proceed to repropose it. Conversely, if the leader lacks the requisite payload, it will initiate a broadcast in the form of $\langle Req, v+1,  U \rangle).$ Nodes possessing the payload for $U$ will respond by transmitting the necessary payload. Given that the block $b$ has achieved commitment by at least one node, signifying that a minimum of $f+1$ honest nodes have encountered it and cast their votes in favor, the new leader is guaranteed to receive the essential payload.

It is worth emphasizing that a potential scenario exists in which the $highU$ contained within the $TC$ has not attained commitment by any node. In such instances, a node that has not observed the block associated with the $highU$ and consequently refrained from casting a vote for it will issue a negative response following the format $\langle N, v+1, U \rangle$ to the leader.

The leader is responsible for collecting and consolidating these negative responses into a No-Commit Certificate ($NC$), which is subsequently integrated into the block. The leader proceeds to extend $highQC$ by proposing a block. This mechanism serves as a  solution to scenario in which a block associated with $highU$ does not receive commitment from any honest node, thereby mandating the inclusion of a No-Commit Certificate ($NC$) within the $TC$. The $TC$ is subsequently embedded in the block proposed by the new leader. 

       \textbf{Scenario 3: Equivocation}

      In the event of equivocation, where a Byzantine leader proposes multiple blocks for the same view, the VBFT protocol introduces mechanisms to address potential safety concerns. If one of the proposed blocks attains commitment by a majority of $f+1$ honest nodes, the equivocation does not lead to the revocation of the block committed by these $f+1$ honest nodes. This safeguard stems from the fact that in a network with $n=3f+1$ nodes, there exist insufficient nodes ($n-(f+1)=2f$) to form a quorum ($2f+1$) capable of revoking the committed block, thus preserving the property of R-Safety.

However, should at most $f$ honest nodes commit one of the equivocated blocks, a scenario arises in which the block committed by at most $f$ nodes may face revocation. To elucidate the revocation process, consider a hypothetical situation involving two equivocated blocks, namely $b_1$ and $b_2$, both proposed by the Byzantine leader. In this context, assume that $b_1$ has garnered commitment from $f$ honest nodes. During the view change, the new leader receives $\textsc{timeout}$ messages from $2f+1$ nodes, with the assumption that the Quorum Certificate ($QC$) for $b_1$ does not feature among the $2f+1$ $\textsc{timeout}$ messages.

The $\textsc{timeout}$ messages received by the new leader may include references to $U$ values for both blocks, namely $b_1$ ($U_1$) and $b_2$ ($U_2$). In the ensuing decision-making process, if the new leader opts to recover the payload for $U_1$ and subsequently repropose it, the safety property (S-safety) remains intact. However, should the new leader choose to extend $U_2$, this action necessitates the revocation of $b_1$ by up to $f$ nodes that had previously committed it. Consequently, S-safety for the block proposed by the Byzantine leader is compromised.

It is essential to emphasize that in scenarios where equivocation occurs and at most $f$ honest nodes have committed a block, clients have not received the requisite $2f+1$ $Reply$ messages. Consequently, clients do not perceive the transactions within the block ($b_1$) as committed. Thus, revocation in this context does not result in double-spending, as a client can only obtain $2f+1$ $Reply$ messages for a transaction when its subsequent block attains commitment by at least $f+1$ honest nodes. As previously mentioned, if a block is committed by $f+1$ honest nodes, it remains immune to revocation, as affirmed by the safety proof. Therefore, to deter Byzantine nodes from engaging in equivocation, VBFT incorporates the practice of blacklisting equivocating leaders, as observed in related works such as \cite{jalalzai2020hermes,RBFT}, or the potential imposition of stake penalties when a Proof-of-Stake consensus mechanism is employed.

Furthermore, it is important to note that even in cases where a Byzantine leader resorts to equivocation, the revocation is confined to the leader's own block, which has not been acknowledged as committed by clients. This limited revocation does not adversely impact the chain quality and integrity of the blockchain.

\textbf{View Change Optimization (Computation)}
In classic BFT protocols such as SBFT \cite{SBFT}, Castro et al. \cite{Castro:1999:PBF:296806.296824}, and Jalal-Window \cite{Jalal-Window}, view changes can incur a quadratic cost in terms of signature verification within the $TC$ (Timeout Certificate).

Recent proposals have emerged to enhance the efficiency of message and computational costs associated with signatures during view changes, as exemplified in No-Commit Proofs \cite{No-Commit-Proofs} and the work by Jalalzai et al. \cite{jalalzai2021fasthotstuff}. Among these solutions, we have chosen to adopt the approach presented in \cite{jalalzai2021fasthotstuff} due to its simplicity and ease of implementation.

In Fast-HotStuff \cite{jalalzai2021fasthotstuff}, the authors demonstrate that the previously quadratic cost of signature verification can be reduced to linear. This reduction is achieved by verifying the aggregated signature of the $TC$ in conjunction with the aggregated signature of the $\textsc{timeout}$ message associated with the highest $QC$. Such verification is adequate to ensure the validity of the $TC$.

The verification of the aggregated signature of the $TC$ serves to confirm the presence of the $\textsc{timeout}$ message from at least $2f+1$ distinct nodes. Concurrently, the verification of the signature of the latest $QC$ ensures the validity of both the latest $QC$ and potential latest $U$ messages. Consequently, individual nodes are relieved from the necessity of verifying aggregated signatures from the remaining $QC$s.

While this optimization does not directly address the message complexity during view changes, it is worth noting that in practice, the cumulative size of view change messages is inconsequential when compared to the block size. Therefore, their impact on system performance remains negligible, barring scenarios where either the network size is exceptionally large or the block size is exceedingly small, in which case the view change message complexity may begin to affect performance\footnote{In instances of extensive network sizes or diminutive block sizes, the impact of view change message complexity on performance may become more pronounced.}.

\section{Proof of Correctness}
\label{Section: Proof}
In this section, we provide proof of safety and liveness for the VBFT protocol. We consider $S_h$ as the set of honest nodes in the network and $2f+1 \leq |S_h| \leq n$.
\subsection{Safety}
As stated before,  VBFT satisfies S-safety during the normal protocol. VBFT also satisfies S-safety during a view change in the absence of equivocation \textcolor{black}{or when the leader is honest}. Being S-safe also implies that the protocol is R-Safe.
The VBFT protocol may fall back to R-safety \textcolor{black}{with a Byzantine leader.} 
Below we provide lemmas related to the VBFT safety.

\begin{lemma}
\label{Lemma:S-safe-Normal}
VBFT ensures S-Safety during the happy path execution of the protocol.
\end{lemma}

\begin{proof}
  Let us assume, for the sake of contradiction, that there exist two committed blocks, $b_1$ and $b_2$, such that $b_1.parent = b_2.parent$. Nodes in a set $S_1$ have voted for $b_1$, and nodes in $S_2$ have voted for $b_2$.  Both $S_1$ and $S_2$ are of size $2f+1$. Therefore there is at least one honest node $j \in S_1 \cap S_2$ that have voted both for $b_1$ and $b_2$.
Additionally, assume that block $b_1$ is proposed during view $v$ and block $b_2$ is committed during view $v'$.

  We will consider the following cases:

  \textbf{Case 1: $v' == v$}

  Since $b_1$ is committed during view $v$, it implies that at least $2f+1$ nodes (in $S_1$) have voted for $b_1$. Consequently, these nodes have incremented their views to $v+1$ as part of the protocol (Algorithm \ref{Algorithm: Normal-Mode}, line 32). Since $v' == v$, any block $b_2$ proposed in the same view $v$ will
  not be able to collect $2f+1$ votes. This is due to the reason that there is at least one honest node $j \in S_1 \cap S_2$ 
  that has already voted for $b_1$ and will  reject $b_2$ as $b_{2}.v > curView==false$ according to the line 30 of Algorithm \ref{Algorithm: Normal-Mode}.

  \textbf{Case 2: $v' < v$}

  This case follows a similar argument to Case 3.1. If $v' < v$, the same reasoning applies, and any block proposed in the earlier view $v'$ will be rejected because $v' < curView$.

  \textbf{Case 3: $v' > v$ }

  During the happy path, a block $b_2$ is proposed in view $v'$ and must carry a Quorum Certificate (QC) $qc_2$ such that $qc_2.v + 1 = b_2.v$ and $qc_2.v \geq b_1.v$. If this condition is not met, the block $b_2$ will be rejected by the SafetyCheck() condition  line 17 of Algorithm \ref{Algorithm:Utilities} in node $j$, because it cannot extend highQC.block ($b_1$).  

  Therefore, it is not possible for block $b_2$ to carry a QC $qc_2$ with $qc_2.v < b.v$. Hence, during the happy path, VBFT is S-Safe.

  This concludes the proof, demonstrating that VBFT maintains S-safety throughout the happy path execution of the protocol.
\end{proof}

  


\begin{lemma}
    VBFT is S-Safe during view change for a correct/honest leader.
         \label{Lemma: S-safety-ViewChange-NoEq}

\end{lemma}

\begin{proof}
  A block $b$ is committed by a single honest node $i$  during $v$. This means a set of nodes $S_1 \geq 2f+1$ have voted for the block $b$ during the view $v-1$.  Other nodes did not receive the block and triggered the unhappy path by timing out and switching to view $v$.
  During the unhappy path, the new leader collects $\textsc{timeout}$ messages from another set of $2f+1$ nodes ($S_2$) into $\textsc{TC}$ message. There is no guarantee that node $i$'s $\textsc{timeout}$ message is included in the $\textsc{TC}$ message built from $\textsc{timeout}$ message of nodes in $S_2$. If  node $i$'s $QC$ for block $b$ from its $\textsc{timeout}$ message is included in the $\textsc{new-view}$ message, then the new leader will simply propose a block that extends the block $b$ (Algorithm \ref{Algorithm: Normal-Mode} lines $18-22$). But in case the $QC$ for block $b$ is not included in $\textsc{new-view}$, then block $b$ can be recovered using information in the $U$ field of the $\textsc{new-view}$ messages in $\textsc{new-view}$ message (Algorithm \ref{Algorithm: Normal-Mode} line $9$, lines $13-15$).
 Each node in $S_2$ has included the latest $U$ (block header it has voted for) in its $\textsc{timeout}$ message it has sent to the leader. The leader has aggregated these messages into $\textsc{TC}$ message. Since $n=3f+1$, $S_1 \cap S_2 \geq f+1$. This means there is at least one honest node $j$ in the intersection of $S_1$ (that has voted for the block $b$) and $S_2$ (has its $\textsc{timeout}$ message $timeout_j \in \textsc{TC}$ message such that  $timeout_j.U.parent=highQC.block\footnote{$highQC$ is the $QC$ with highest sequence (latest $QC$) in the $TC$. $TC$ is the set of $QC$s in the $\textsc{new-view}$ message.}$). 
  Therefore, the new leader can retrieve the payload for $U$ (block $b$'s header) and repropose it Algorithm \ref{Algorithm: Normal-Mode}. 
\end{proof}




\begin{lemma}
In the VBFT consensus protocol, if a block is revoked due to equivocation and initially committed by at most $f$ honest nodes, then no client will consider it as committed.
\end{lemma}

\begin{proof}
We prove this lemma by considering two fundamental aspects of VBFT's safety mechanisms.

\textbf{Insufficient Votes for Revocation:} VBFT achieves R-safety by requiring the commitment of a block by at least $f+1$ honest nodes before it becomes resistant to revocation. In other words, a block cannot be revoked unless a significant portion of the network has already committed to it. This ensures that blocks committed by a minority of nodes remain vulnerable to revocation.

\textbf{Insufficient Replies for Revoked Blocks:} When a block is revoked due to equivocation, it implies that the block proposer has acted maliciously by contradicting their prior proposal. If such a block is initially committed by at most $f$ honest nodes, it will not receive sufficient reply messages from the network.

Consider the scenario where a block is committed by at most $f$ honest nodes and subsequently revoked due to equivocation by the same block proposer. In this case, there won't be enough honest nodes to generate a quorum of reply messages confirming the block's validity. The majority of the network, including the honest nodes, would recognize the equivocation and revoke their support for the conflicting block.

As a result, clients in the network would not receive the expected number of reply messages confirming the block's commitment. Clients typically rely on a threshold number of confirmations from the network to consider a block as committed. Without sufficient confirmations, the clients will not consider the block as committed.

Therefore, the combination of these safety mechanisms in VBFT ensures that if a block is revoked due to equivocation and was initially committed by at most $f$ honest nodes, it will not receive the necessary reply messages from the network for clients to consider it as committed. This guarantees that no client will consider such a block committed, thus demonstrating the R-safety of the VBFT consensus protocol for Byzantine leaders.
\end{proof}

\begin{lemma}
In the VBFT consensus protocol, liveness is guaranteed, ensuring that the protocol makes progress eventually and does not stall indefinitely.

\begin{proof}
We establish the liveness properties of VBFT through a series of techniques and considerations:

\textbf{Exponential Back-Off Timer:} VBFT incorporates an exponential back-off timer PBFT \cite{Castro:1999:PBF:296806.296824} for leader selection during view changes. This timer doubles its timeout period after each view change, providing the next leader with ample time to reach a decision on a request or block. This mechanism ensures that the protocol maintains progress and does not stall indefinitely.

\textbf{Timeout Messages:} To further enhance liveness, nodes in VBFT broadcast $\textsc{timeout}$ messages under specific conditions. If a node receives $f+1$ $\textsc{timeout}$ messages, each from a distinct node, 
it broadcasts a $\textsc{timeout}$ message.
This guarantees that at least one of the $\textsc{timeout}$ messages is from an honest node. Moreover, a node will not broadcast a $\textsc{timeout}$ message and trigger a view change if it receives at most $f$ $\textsc{timeout}$ messages. This restriction prevents the protocol from entering an indefinite state of view changes due to the unhappy path.

\textbf{Leader Selection:} In cases where leaders are selected in a round-robin manner, after at most $f$ Byzantine leaders, an honest leader will be chosen. This ensures that progress will be made eventually. Similarly, if leaders are selected randomly, the probability of selecting a Byzantine leader is $P_b=1/3$. Treating this as a Bernoulli trial, the probability of a bad event (selecting a Byzantine leader) occurring for $k$ consecutive views is $P_b^k$. As $k$ increases, $P_b^k$ rapidly approaches $0$, guaranteeing that an honest node will eventually be selected as the leader, allowing the protocol to make progress.

Hence, VBFT guarantees liveness by employing these techniques and ensuring that progress is achieved within a bounded time frame, preventing indefinite stalls in the protocol.
\end{proof}

\end{lemma}

\section{Conclusion}
\label{Section: Conclusion}
       In this paper, we have presented the VBFT consensus protocol, which achieves consensus during normal protocol operation in just two communication steps. First, we have shown that the previous optimal bound for fast BFT can be improved from $f \leq \frac{n+1}{5}$ to $f \leq \frac{n-1}{3}$ without using any trusted hardware. 
       Second, we have shown that VBFT guarantees strong safety for honest leaders and enjoys efficient view change. 
       Whereas, VBFT achieves  lower latency as compared to  HotStuff in the happy path and compared to PBFT-TE in the unhappy path.


\bibliographystyle{IEEEtran}
\bibliography{new}

\end{document}